\documentclass[11  pt, oneside]{article}   	
\usepackage[margin=1.25in,lmargin=0.75in,rmargin=1.75in,bmargin=0.75in,tmargin=0.75in]{geometry}      
\usepackage{graphicx}
\usepackage{booktabs} 
\usepackage{listings}
\usepackage{amsmath}
\usepackage{amsfonts}
\usepackage{amssymb}
\usepackage{graphicx}
\usepackage{xspace}

\newcommand{\vol}{\textbf{vol}}
\newcommand{\cut}{\textbf{cut}}

\providecommand{\mat}[1]{\boldsymbol{\mathrm{#1}}}%
\renewcommand{\vec}[1]{\boldsymbol{\mathrm{#1}}}

\DeclareMathOperator{\argmin}{argmin}
\providecommand{\subjectto}{\ensuremath{\text{subject to}}}

\providecommand{\mA}{\ensuremath{\mat{A}}}

\providecommand{\vb}{\ensuremath{\vec{b}}}
\providecommand{\vc}{\ensuremath{\vec{c}}}

\providecommand{\vx}{\ensuremath{\vec{x}}}

\newcommand{\CC}{\ensuremath{\mathcal{C}}\xspace}

\newcommand{\lcc}{\textsc{LambdaCC}\xspace}

\usepackage{amsthm}
\newtheorem{theorem}{Theorem}

\newcommand{\tony}[1]{}
\renewcommand{\tony}[1]{{\textcolor{blue}{[{#1}\ --\ AIW]}}}

\usepackage{amsthm}
\usepackage{algorithm}
\usepackage[noend]{algpseudocode}
\usepackage{algorithmicx}

\newcommand{\pp}{\ensuremath{\mathcal{P}}}
\algnewcommand\algorithmicswitch{\textbf{switch}}
\algnewcommand\algorithmiccase{\textbf{case}}
\newcommand{\qq}{\mathcal{Q}}
\algdef{SE}[SWITCH]{Switch}{EndSwitch}[1]{\algorithmicswitch\ #1\ \algorithmicdo}{\algorithmicend\ \algorithmicswitch}
\algdef{SE}[CASE]{Case}{EndCase}[1]{\algorithmiccase\ #1}{\algorithmicend\ \algorithmiccase}
\algtext*{EndSwitch}
\algtext*{EndCase}

\usepackage{booktabs} 
\usepackage{enumitem}
\usepackage{balance}
\usepackage{subfig}
\usepackage[T1]{fontenc}
\usepackage[utf8]{inputenc}
\usepackage{authblk}

\usepackage{url}

\begin{document}
	
\title{Learning Resolution Parameters for Graph Clustering}

\author[a]{Nate Veldt}
\author[b]{David F. Gleich} 
\author[c]{Anthony Wirth}

\affil[a]{Purdue University Mathematics Department}
\affil[b]{Purdue University Computer Science Department}
\affil[c]{The University of Melbourne, Computing and Information Systems School}
\date{}
\maketitle

	\begin{abstract}
		Finding clusters of well-connected nodes in a graph is an extensively studied problem in graph-based data analysis. Because of its many applications, a large number of distinct graph clustering objective functions and algorithms have already been proposed and analyzed. 
		To aid practitioners in determining the best clustering approach to use in different applications, we present new techniques
		for automatically learning how to set clustering resolution parameters. These parameters control the size and structure
		of communities that are formed by optimizing a generalized objective function. We begin by formalizing the notion of a
		parameter fitness function, which measures how well a fixed input clustering 
		approximately solves a generalized clustering objective for a specific resolution parameter value.
		Under reasonable assumptions, which suit two key graph clustering applications,
		such a parameter fitness function can be efficiently minimized using a bisection-like method, yielding a resolution
		parameter that fits well with the example clustering. We view our framework as a type of \emph{single-shot} hyperparameter
		tuning, as we are able to learn a good resolution parameter with just a single example. Our general approach can be applied to learn resolution parameters for both local and global graph clustering objectives. We demonstrate its utility in several experiments on real-world data where it is helpful to learn resolution parameters from a given example clustering.
	\end{abstract}

\section{Introduction}
Partitioning a collection of items into groups of similar items -- that is, clustering -- is a fundamental computational task.
So commonly applied, there is a large and still-growing suite of objective functions, algorithms,
and techniques for identifying \emph{good} clusters.
One powerful mathematical model for clustering is the graph, comprising nodes and (undirected) edges. For a broad overview of graph clustering, refer to any one of a number of surveys~\cite{schaeffer2007graphclustering,fortunato2010,fortunato2016,porter2009communities}.
Nearly all clustering approaches favor clusters with high internal edge density and a low external edge density. A related, but not identical, notion is that a good cluster is a set of nodes with a small cut (i.e., few edges leaving the set), and a nontrivial size (e.g., a large number of nodes, or many internal edges). 

Although most clustering approaches follow these general principles, there are many different ways to formalize such
goals mathematically. In practice, this array of objective functions yields a large variety of
different output clusterings.
Indeed, many existing (theoretical)
approaches to graph clustering assume that the user knows \emph{a priori} which objective function
is appropriate for their context or job.
The main design task, leading to a practical solution, is then to develop good algorithms that exactly, or approximately, optimize the objective. However, we propose that it is more natural to assume that the user starts with some a priori knowledge about the desired \emph{structure} of clusters in a given application domain.
More specifically, they can provide at least \emph{one} example of what a \emph{good} clustering should look like. The revised goal is to find an objective function whose optimization yields the desired type of output.

\paragraph{Our approach}
In this article, we show how to bootstrap from this one quality clustering to learn the appropriate objective function, chosen from a parameterized family of objective functions.
The efficiency of our technique relies on the clustering objective function being linear, with linear constraints, tuned by a single parameter $\beta$.
Given this type of objective, we formalize the notion of a parameter fitness function, which relies on a fixed example clustering~$C_x$ of a network, and takes a parameter $\beta$ as input. The fitness function computes the ratio between the objective score of~$C_x$, when~$\beta$ is chosen as the input parameter, and a lower bound on the optimal clustering objective for that~$\beta$. Specifically, we use a concave piecewise linear lower bound over a wide family of what we refer to as \emph{relaxed clusterings}. Minimizing the fitness function produces a parameter (and a corresponding objective function) that~$C_x$ exactly or at least \emph{nearly} optimizes. Thus the aim, which can be realized via a guided binary search, is to identify the parameter setting in which~$C_x$ most stands out. In the remainder of this introduction, we flesh out the context for our technique.

\paragraph{Parameters}
There are multiple families of objective functions
whose members are specified by a tunable resolution parameter.
This parameter controls the size and structure of clusters that are formed by optimizing the objective.
Key examples of such \emph{generalized clustering objective functions} include the Hamiltonian objective studied by Reichardt and Bornholdt~\cite{reichardt2006statistical}, the {stability} objective of Delvenne
et al.~\cite{delvenne2010stability}, and a multi-resolution variant of the map equation~\cite{schaub2012multiscale}.

In this manuscript we focus on a related clustering framework that we developed in previous work~\cite{Veldt:2018:CCF:3178876.3186110}, based on correlation clustering~\cite{Bansal2004correlation}. This framework is named \lcc, after its resolution
parameter~$\lambda$, which implicitly controls both the internal edge density as well as the cut sparsity of nodes formed by minimizing the objective. Furthermore, \lcc generalizes several well-studied objectives such as modularity
clustering, sparsest cut, normalized cut, and cluster deletion. All of these objectives can be viewed as special cases
of the objective for appropriate settings of~$\lambda$. 

\paragraph{Global and local}
The above objectives are specifically designed for global clustering, in which the goal is to find a multi-cluster partitioning of an input graph. Local clustering objectives relying on resolution parameters also exist; these focus on finding a single cluster in a localized region of a graph. Flow-based methods such as FlowImprove~\cite{AndersenLang2008}, LocalImprove~\cite{OrecchiaZhu2014}, and SimpleLocal~\cite{VeldtGleichMahoney2016} fit into this category.
These methods repeatedly solve minimum $s$-$t$ cut problems for different values of a parameter~$\alpha$, in
order to minimize a ratio-style objective related to a cluster quality measure called \emph{conductance}.
This~$\alpha$ can be viewed as a resolution parameter that balances a trade-off between forming a cluster with a small cut, and forming a cluster with a large overlap with a seed set in the graph.

Given the unifying power and versatility of generalized clustering objective functions, the challenge of finding the right clustering technique for a specific application can often be reduced to finding an appropriate resolution parameter. However, very little work has addressed how to set these parameters in practice, in particular to capture the specific clustering structure present in a certain application domain. In the past, solving generalized objective functions for a range of resolution parameters has been used to detect hierarchical clusterings in a network~\cite{reichardt2006statistical}, or as a way to identify~\emph{stable} clusterings, which consistently optimize the objective over a range of parameter values~\cite{delvenne2010stability,jeub2017multiresolution,schaub2012multiscale}. While both are important applications of resolution-based clustering, the ability to detect a specified type of clustering structure is important regardless of a clustering's stability or the hierarchical structure of a network. Finally, while tuning hyperparameters is a standard procedure in the broader machine learning literature, most existing approaches are not specifically designed for tuning graph clustering resolution parameters. Furthermore hyperparameter tuning techniques typically rely on performing cross validation over a large number of training examples. We are concerned with learning good resolution parameters from a \emph{single} example clustering that represents a meaningful partitioning in a certain application domain.

\paragraph{Our Contributions}
In this paper we develop an approach for learning how to set resolution parameters for both local and global graph clustering problems. Our results for global graph clustering rely heavily on the \lcc framework we developed in past work~\cite{Veldt:2018:CCF:3178876.3186110}. We begin by formally defining a parameter fitness function for a given clustering. We then prove that under reasonable assumptions on the input clustering and clustering objective function used, we can find the minimizer of such a fitness function to within arbitrary precision using a simple bisection-like method. Our approach can be viewed as a type of \emph{single-shot} hyperparameter tuning, as we are able to learn an appropriate setting of a resolution parameter when given a single example clustering. We display the utility of our approach in several local and global graph clustering experiments. Our approach allows us to obtain improved community detection results on synthetic and real-world networks. We also show how our method can be used to measure the correlation between metadata attributes and community structure in social networks.

\section{Graph Clustering Background}
This section reviews the global \lcc~\cite{Veldt:2018:CCF:3178876.3186110} clustering objective and a local clustering objective that is based on regionally biased minimum cut computations~\cite{OrecchiaZhu2014,AndersenLang2008}. While there do exist many other objectives for local and global clustering, we focus on these two as they both rely crucially on resolution parameters. 

\paragraph{Basic Notation}
In this paper we consider unweighted and undirected graphs $G = (V,E)$, though many of the ideas can be extended to weighted graphs. Global graph clustering separates $G$ into disjoint sets of nodes so that every node belongs to exactly one cluster. For local clustering, one is additionally given a set of reference or seed nodes $R \subset V$ and the objective is to find a good cluster that shares a nontrivial overlap with $R$. The degree of a node $i \in V$ is the number of edges incident to it; we denote this by $d_i$. The volume of a set $S \subseteq V$ is given by $\vol(S) = \sum_{i \in S} d_i$ and $\cut(S)$ measures how many edges cross from $S$ to its complement set $\bar{S} = V\backslash S$. Further notation will be presented as needed in the paper.

\subsection{Global Clustering with \lcc}
\label{sec:lamcc}
The \lcc objective is a special case of correlation clustering (CC)~\cite{Bansal2004correlation}, a framework for partitioning signed graphs. In correlation clustering, each pair of nodes $(u,v)$ in a signed graph is associated with either a positive edge or a negative edge, as well as a nonnegative edge weight $e_{uv}$ indicating the strength of the relationship between $u$ and $v$. Given this input, the goal is to produce a clustering which minimizes the weight of \emph{disagreements} or \emph{mistakes}, which occur when positive edges are placed between clusters or negative edges are placed inside clusters.

The \lcc framework takes an unsigned graph $G = (V,E)$, a resolution parameter $\lambda \in (0,1)$, and node weights $w_u$ for each $u \in V$. It converts this input into a signed graph over which the correlation clustering objective can then be minimized. The signed graph $\tilde{G} = (V,E^+, E^-)$ is constructed as follows: for every $(u,v) \in E$, if $1 - \lambda w_u w_v \geq 0$, form a positive edge $(u,v)$ in $\tilde{G}$, otherwise form a negative edge. In either case, the weight of this edge is $e_{uv} = |1 - \lambda w_u w_v|$. For every non-edge in the original graph ($(u,v) \notin E$), form a negative edge $(u,v) \in E^-$ in $\tilde{G}$ with weight $e_{uv} = \lambda w_u w_v$. The \lcc objective function then corresponds to the correlation clustering objective applied to $\tilde{G}$:
%
\begin{equation}
\label{eq:lcc}
{ \min \,\, \sum_{(u,v) \in E^+} e_{uv}(1-\delta_{uv}) + \sum_{(i,j) \in E^-} e_{uv} \delta_{uv}, }
\end{equation}
where $\delta_{uv}$ is a zero-one indicator function which encodes whether a clustering has placed nodes $u,v$ together ($\delta_{uv} = 1$), or apart $(\delta_{uv} = 0$). 
There are two main choices for node weights: $w_u = 1$ for all $u \in V$ is the \emph{standard} \lcc objective. For this simple case, we note that a node pair $(u,v)$ which defines an edge in $G$ will always correspond to a positive edge in $G$. In some applications it is useful to consider a \emph{degree-weighted} version in which $w_u = d_u$. In this case, if $\lambda \leq 1/(d_{max}^2)$ then we can still guarantee that $E = E^+$. However, for larger values of $\lambda$ it may be possible that an edge in $G$ gets mapped to a negative edge in $\tilde{G}$.

As a generalization of standard unweighted correlation clustering, \lcc is NP-hard, though many approximation algorithms and heuristics for correlation clustering have been developed in practice~\cite{ailon2008aggregating,Bansal2004correlation,charikar2005clustering,chawla2015near}. In our previous work~\cite{Veldt:2018:CCF:3178876.3186110}, we showed that a 3-approximation for standard \lcc can be obtained for any $\lambda \geq 1/2$ by rounding the following LP relaxation of objective~\eqref{eq:lcc}:
\begin{equation}
\label{eq:lccLP}
\begin{array}{lll} \text{minimize} & \sum_{(u,v) \in E} (1-\lambda) x_{uv} \,\, + & \sum_{(u,v) \notin E} \lambda (1-x_{uv})\\  \subjectto
& x_{uv} \leq x_{uw} + x_{vw} & \text{for all $u,v,w$} \\
& 0 \leq x_{uv} \leq 1 & \text{for all $u < v$.}
\end{array}
\end{equation}
Furthermore, even when a priori approximations are not guaranteed, solving the LP relaxation can be a very useful way to obtain empirical lower bounds for the objective in polynomial time. In follow up work~\cite{Veldt2018ccgen}, we provided improved approximations for $\lambda < 1/2$ based on rounding the LP, but noted an $\Omega(\log n)$ integrality gap for some small value of $\lambda$.


\paragraph{Equivalence Results} \lcc generalizes and unifies a large number of other clustering approaches. When $\lambda = 1/(2|E|)$, the degree-weighted version is equivalent to the popular maximum modularity clustering objective~\cite{newman2004modularity,newman2006finding}. Standard \lcc interpolates between the sparsest cut objective for a graph-dependent small value of $\lambda$, and the cluster deletion problem when $\lambda > |E|/(1+|E|)$. Given its relationship to modularity, \lcc is known to also be related to the stochastic block model~\cite{newman2013equivalence} and a multi-cluster normalized cut objective~\cite{yuModularityncut2}.

\subsection{Local Clustering Objectives}
We next consider a class of clustering objectives that share some similarities with~\eqref{eq:lcc}, but are designed for finding a single local cluster in a specific region of a large graph. With the input graph $G = (V,E)$, we additionally identify a set of seed or \emph{reference} nodes~$R$ around which we wish to form a good community. One common measure for the ``goodness'' of a cluster $S$ is the conductance objective:
\begin{equation}
\label{cond}
\phi(S) = {\cut(S)}/ ({\min\{ \vol(S), \vol(\bar{S})\} }),
\end{equation}
which is small when $S$ is connected very well internally but shares few edges with $\bar{S}$. A number of graph clustering algorithms have been designed to minimize local variants of~\eqref{cond}. These substitute the denominator of~\eqref{cond} with a measure of the overlap between an output cluster $S$ and the reference set $R$. One such objective is the following local conductance measure:
\begin{equation}
\label{localcond}
\phi_R(S) = \frac{\cut(S)}{\vol(R \cap S) - \varepsilon \vol(\bar{R} \cap S)},
\end{equation}
which is minimized over all sets $S$ such that the denominator of $\phi_R(S)$ is positive. This objective includes a \emph{locality} parameter $\varepsilon$ that controls how much overlap there should be between the seed set and output cluster. For a general overview of this clustering paradigm and its relationship to spectral and random-walk based techniques, we refer the reader to the work of Fountoulakis et al.~\cite{fountoulakis2017Optimization}. Specific algorithms which minimize variants of~\eqref{localcond} include FlowImprove~\cite{AndersenLang2008}, which always uses parameter $\varepsilon = \vol(R)/\vol(\bar{R})$, and LocalImprove~\cite{OrecchiaZhu2014} and SimpleLocal~\cite{VeldtGleichMahoney2016}, both of which choose larger values of $\varepsilon$ in order to keep computations more local. In the extreme case where we consider $\varepsilon = \infty$, the problem reduces to finding the minimum conductance \emph{subset} of a reference set $R$, which can be accomplished by the Minimum Quotient Improvement (MQI) algorithm of Lang and Rao~\cite{LangRao2004}.

Objective~\eqref{localcond} can be efficiently minimized by repeatedly solving a minimum $s$-$t$ cut problem on an auxiliary graph constructed from $G$, which introduces a sink node $s$ attached to nodes in $R$, and a source node $t$ attached to nodes in $\bar{R} = V \backslash R$. Edges are weighted with respect to the locality parameter $\varepsilon$ and another parameter $\alpha$. In order to detect whether there exists some set $S$ with $\phi_R(S) \leq \alpha$, one can solve a local clustering objective corresponding to the minimum $s$-$t$ cut objective on the auxiliary graph. We refer to this simply as the \emph{local flow clustering objective}:
\begin{equation}
\label{stcut}
\min \,\, f_\alpha(S) = \cut(S) + \alpha \vol(R \cap \bar{S}) + \alpha \varepsilon \vol(\bar{R} \cap S).
\end{equation}
If the set $S$ minimizing $f_\alpha$ satisfies $f_\alpha(S) < \alpha \vol(R)$, then rearranging terms one can show that $\phi_R(S) < \alpha$. Thus, by performing binary search over $\alpha$ or repeatedly solving~\eqref{stcut} for smaller and smaller $\alpha$, one can minimize the local conductance measure~\eqref{localcond}. 
%

Previous research has largely treated~$\alpha$ as a temporary parameter used in one step of a larger algorithm seeking to minimize~\eqref{localcond}. Algorithms which minimize~\eqref{localcond} do so by finding the smallest~$\alpha$ such that the minimum of~\eqref{stcut} is~$\alpha \vol(R)$. We depart from this approach by instead treating~$\alpha$ as a tunable resolution parameter for balancing two conflicting goals: finding clusters with a small cut, and finding clusters that have a large overlap with the seed set~$R$. In the case where~$\varepsilon$ is treated as infinitely large and we are simply looking for subsets of a seed set~$R$ satisfying $\vol(R) \leq \vol(\bar{R})$, then in effect we are trying to solve the optimization problem:
\begin{equation}
\label{mqiobj}
\min \cut(S) - \alpha \vol(S) + \alpha\vol(R)\,\, \text{ such that $S \subseteq R$}.
\end{equation}
This goal is related to, but ultimately should be contrasted with, the goal of minimizing the ratio $\cut(S)/\vol(S)$. The objectives are similar in that they both tend to prefer sets with small cut and large volume. We argue that treating $\alpha$ as a tunable parameter is in fact more versatile than simply minimizing the ratio score. In multiple applications it may be useful to find clusters with small cut and large volume, but different applications may put a different weight on each aspect of the objective.
We observe that $\varepsilon$ also plays an important role in the size and structure of the output community when it is less than $\infty$. For simplicity, in this paper we can treat this as a fixed constant, and in our experimental section we simply focus on objective~\eqref{mqiobj}.


\subsection{Parametric Linear Programs}
\label{sec:paraLP}
Before moving on we provide key background on parametric linear programming which will be important in our theoretical results. A standard linear program is a problem of the form
\begin{equation}
\label{eq:lp}
\min_{\vx} \,\, \vc^T \vx \,\, \text{such that $\mA \vx \leq \vb$}
\end{equation}
where $\vc, \vb$ are vectors and $\textbf{A}$ is a constraint matrix. A \emph{parametric} linear program is a related problem of the form
\begin{equation}
\label{eq:plp}
\min_{\vx} \,\, \vc^T \vx + \beta(\Delta\vc)^T\vx \,\, \text{such that $\mA \vx \leq \vb$}
\end{equation}
where $\Delta \vc$ is another vector of the same length as $\vc$ and $\beta$ is a parameter controlling the difference between~\eqref{eq:lp} and~\eqref{eq:plp}. We state a well-known result about the solutions of~\eqref{eq:plp} for different $\beta$. This result is not new; it follows directly from Proposition 2.3b from~\cite{adler1992geometric}.
\begin{theorem}
	\label{thm:concave}
	Let~$L(\beta)$ be the minimum of~\eqref{eq:plp} for a fixed~$\beta$. If we are given bounds~$a$ and~$b$ such that~$L(\beta) \in \mathbb{R}$ for all $\beta \in [a,b]$, then~$L$ is a piecewise linear and concave function in~$\beta$ over this interval.
\end{theorem}

\paragraph{Parametric LPs in Graph Clustering Applications}
In our work it is significant to note that the linear programming relaxation of \lcc is a parametric linear program in $\lambda$. Furthermore, the local flow clustering objective can be cast as a parametric linear program in $\alpha$, since this objective corresponds simply to a special case of the minimum $s$-$t$ cut problem, which can be cast as an LP. 

\subsection{Related Work}
Our work builds on previous results that introduced generalized objective functions with resolution parameters, including the Hamiltonian objective~\cite{reichardt2006statistical}, clustering stability~\cite{delvenne2010stability}, a multiscale variant of the map equation~\cite{schaub2012multiscale}, and the \lcc framework~\cite{Veldt:2018:CCF:3178876.3186110}. Recently Jeub et al.~\cite{jeub2017multiresolution} introduced a technique for sampling values of a resolution parameter and applying hierarchical consensus clustering techniques. Our work on learning clustering resolution parameters differs from theirs in that we do not aim to provide hierarchical clusterings of a network. Instead we assume that there is a known fixed clustering, for which we wish to learn a single specific resolution parameter. 

There exist many techniques for localized community detection based on seed set expansion. Among numerous others, these include spectral and random-walk based methods~\cite{Spielman-2013-local,andersen2006-local}, flow-based methods~\cite{LangRao2004,AndersenLang2008,OrecchiaZhu2014,VeldtGleichMahoney2016}, and other approaches which perform diffusions from a set of seed nodes and round embeddings via a sweep cut procedure~\cite{pmlr-v70-wang17b,Kloster-2014-hkrelax}. We build on these by interpreting hyperparameters associated with such methods as resolution parameters which can be learned to produce clusters of a certain type.
\section{Theoretical Results}
The major theoretical contribution of our work is a new framework for learning clustering resolution parameters based on minimizing a parameter fitness function for a given example clustering. We present results for a generic clustering objective and fitness function, and later show how to apply our results to \lcc and local flow clustering.

\subsection{Problem Formulation}
\label{goals}
Let $\CC$ denote a set of valid clusterings for a graph $G = (V,E)$. We consider a generic clustering objective function $f_\beta: \CC \rightarrow \mathbb{R}_{\geq 0}$ that depends on a resolution parameter $\beta$. The function takes as input a clustering $C \in \CC$, and outputs a nonnegative clustering quality score for~$C$. We assume that smaller values of $f_\beta$ are better. We intentionally allow $f_\beta$ to be very general in order to develop broadly applicable theory. For intuition, one can think of $f_\beta$ as being the \lcc function~\eqref{eq:lcc} with $\beta = \lambda$. Alternatively, one can picture $f_\beta$ to be the local flow objective~\eqref{stcut} with $\beta = \alpha$ and with~$\CC$ representing the set of bipartitions, i.e.\ for any $C \in \CC$, $C = \{S, \bar{S} \}$ for some set $S \subset V$.

Given some objective function $f_\beta$, a standard clustering paradigm is to assume that an appropriate value of $\beta$ has already been chosen, and then the goal is to produce some clustering $C$ that exactly or approximately minimizes $f_\beta$. In our work, we address an inverse question: given an example clustering $C_x$, how do we determine a parameter $\beta$ such that $C_x$ approximately minimizes $f_\beta$? Ideally we would like to solve the following problem:
\begin{equation}
\label{eq:ideal}
\textbf{Goal 1: }\text{ Find $\beta > 0$ such that } f_\beta(C_x) \leq f_\beta(C) \text{ for all $C \in \CC$}.
\end{equation}
In practice, however, $C_x$ may not exactly minimize a generic clustering objective for any choice of resolution parameter. Thus we relax this to a more general and useful goal:
\begin{align}
\label{eq:better}
\textbf{Goal 2: } &\text{ Find the minimum } \Delta \geq 1 \text{ such that  for some $\beta > 0$} \notag \\
 &f_\beta(C_x) \leq \Delta f_\beta(C) \text{ for all $C \in \CC$}.
\end{align}
This second goal is motivated by the study of approximation algorithms for clustering. In effect this asks: if we are given a certain clustering $C_x$, is $C_x$ a good approximation to $f_\beta$ for any choice of $\beta$? Note that this generalizes~\eqref{eq:ideal}: if $\beta$ can be chosen to satisfy Goal 1, then the same $\beta$ will satisfy Goal 2 with $\Delta = 1$. Furthermore, it has the added advantage that, if solved, Goal 2 will produce a value $\Delta$ which communicates how well clusterings like $C_x$ can be detected using variants of the objective function $f_\beta$. If $\Delta$ is near 1, it means that~$f_\beta$ is able to produce similar clusterings for a correct choice of $\beta$, whereas if $\Delta$ is very large this indicates that $C_x$ will be difficult to find even for an optimal $\beta$, and thus a different approach will be necessary for detecting clusterings of this type.

\paragraph{Clustering Relaxations} While Goal 2 is a more reasonable target than Goal 1, it may still be a very challenging problem to solve when objective $f_\beta$ is hard to optimize, e.g., if it is NP-hard. We thus consider one final relaxation that is slightly weaker than~\eqref{eq:better}, but will be more feasible to work with. Let $\hat{\CC}$ denote a superset of $\CC$ which includes not only clusterings for $G$, but also some notion of a relaxed clustering, and let $g_\beta : \hat{\CC} \rightarrow \mathbb{R}_{\geq 0}$ be an objective that assigns a score for every $C \in \hat{\CC}$. Furthermore, assume $g_\beta$ represents a lower bound function for $f_\beta$: $g_\beta(C) \leq f_\beta(C) \text{ for all $\beta$ and all $C \in \CC$}$.
Our consideration of $g_\beta$ is motivated by the fact that many NP-hard clustering objectives permit convex relaxations, which can be optimized in polynomial time over a larger set of relaxed clusterings that contain all valid clusterings of $G$ as a subset. For example, the \lcc objective is NP-hard to optimize for every $\lambda \in (0,1)$, but the linear programming relaxation for every $\lambda$ can be solved in polynomial time, and is defined over relaxed clusterings in which pairs of nodes are assigned distances between 0 and 1. These relaxations can be rounded to produce good approximations to the original NP-hard objective~\cite{charikar2005clustering,chawla2015near}. Since $g_\beta$ is indeed easier to optimize than $f_\beta$, the following goal will be easier to approach but still provide strong guarantees for learning a good value of $\beta$:
\begin{align}
\label{eq:best}
\textbf{Goal 3: } &\text{ Find the minimum } \Delta \geq 1 \text{ such that  for some $\beta > 0$} \notag \\
&f_\beta(C_x) \leq \Delta g_\beta(C) \text{ for all $C \in \hat{\CC}$}. 
\end{align}
If we can solve~\eqref{eq:best}, this still guarantees that $C_x$ is a $\Delta$-approximation to $f_\beta$ for an appropriately chosen $\beta$. For problems where $f_\beta$ is very challenging to optimize, but $g_\beta$ is not, this will be a much more feasible approach. In the next section we will focus on developing theory for addressing Goal 3, though we note that in applying this theory we can still choose $g_\beta = f_\beta$ and therefore instead address the stronger Goal 2 whenever this is feasible. We will take this approach when applying our theory to the local flow objective.

\subsection{Parameter Fitness Function}
We now present a parameter fitness function whose minimization is equivalent to solving~\eqref{eq:best}. Functions $f_\beta$ and $g_\beta$ take a clustering or relaxed clustering as input and output an objective score. However, we wish to view $\beta$ as an input parameter and we treat an example clustering $C_x$ as a fixed input. Thus for convenience we introduce new related functions:
\begin{align}
	F(\beta) & = f_\beta(C_x) \\
	G(\beta) & = \min_{C \in \hat{\CC}} g_\beta(C)
\end{align}
The ratio of these two functions defines the \emph{parameter fitness function} that we seek to minimize:
\begin{equation}
\label{eq:pff}
\mathcal{P}(\beta) = \frac{F(\beta)}{G(\beta)}.
\end{equation}
Observe that this function is always greater than or equal to 1 since $G(\beta) \leq F(\beta)$ for any $\beta$. The minimizer of $\mathcal{P}$ is a resolution parameter $\beta$ that minimizes the ratio between the clustering score of a fixed $C_x$ and a lower bound on $f_\beta$. Thus, by minimizing~\eqref{eq:pff} we achieve Goal 3 in~\eqref{eq:best} with $\Delta = \min_\beta \mathcal{P}(\beta)$. 

In Section~\ref{sec:paraLP}, we noted that the local flow clustering objective can be characterized as a parametric linear program, as can the LP relaxation of \lcc. Furthermore, for a fixed clustering, both objective functions can be viewed as a linear function in terms of their resolution parameter. Motivated by these facts, we present a theorem which characterizes the behavior of the parameter fitness function $\mathcal{P}$ under certain reasonable conditions on the functions $F$ and $G$. In the subsequent section we will use this result to show that $\mathcal{P}$ can be minimized to within arbitrary precision using an efficient bisection-like method.
\begin{theorem}
	\label{thm:pff}
	Assume $F(\beta) = a + b \beta$ for nonzero real numbers $a$ and $b$. Let $G$ be concave and piecewise linear in $\beta$, and assume $F(\beta) \geq G(\beta) \geq 0$ for all $\beta \in [\ell,r]$ where $\ell$ and $r$ are nonnegative lower and upper (i.e.\ left and right) bounds for $\beta$. Then $\pp$ satisfies the following two properties:
	\begin{enumerate}[label=(\alph*)]
		\item If $\beta^- < \beta < \beta^+$, then $\mathcal{P}(\beta)$ cannot be strictly greater than both $\mathcal{P}(\beta^-)$ and $\mathcal{P}(\beta^+)$.
		\item If $\mathcal{P}(\beta^-) = \mathcal{P}(\beta^+)$, then $\mathcal{P}$ achieves its minimum in $[\beta^-,\beta^+]$.
	\end{enumerate}
\end{theorem}
\begin{proof}
	Note that for some $\gamma \in (0,1)$, $\beta = (1-\gamma)\beta^+ + \gamma\beta^-$. By concavity of $G$ and linearity of $F$, we know
	\begin{align*}
	\mathcal{P}(\beta) &= \frac{F((1-\gamma)\beta^+ + \gamma\beta^-)}{G((1-\gamma)\beta^+ + \gamma\beta^-)}
								 \leq  \frac{(1-\gamma)F(\beta^+) + \gamma F(\beta^-)}{(1-\gamma)G(\beta^+) + \gamma G(\beta^-)}\\
								 &\leq \max \left\{ \frac{(1-\gamma)F(\beta^+)}{(1-\gamma)G(\beta^+)}, \frac{\gamma F(\beta^-)}{\gamma G(\beta^-)}   \right\}
								 = \max \left \{\mathcal{P}(\beta^+), \mathcal{P}(\beta^-)  \right\},
	\end{align*}
	which proves the first property. Now assume that $\mathcal{P}(\beta^-) = \mathcal{P}(\beta^+)$. Using property 1, we know as $\beta$ increases from its lower to upper limit, $\mathcal{P}$ cannot increase and then decrease. Thus, either $\mathcal{P}$ attains its minimum on $[\beta^-,\beta^+]$, else $\mathcal{P}$ is a constant for all $\beta \in  [\beta^-,\beta^+]$. If the latter is true, then for some $\beta \in  [\beta^-,\beta^+]$ and some sufficiently small $\epsilon > 0$, $G$ must be linear in the range $(\beta-\epsilon, \beta + \epsilon)$, since we know that $G$ is piecewise linear. Therefore, $G(\beta) = c+ d\beta$ and
	\begin{equation}
	\mathcal{P}(\beta) = (a+b\beta)/(c+d\beta)= constant
	\end{equation}
	for $\beta \in (\beta-\epsilon, \beta + \epsilon)$ and for some $c, d \in \mathbb{R}$. This ratio of linear functions can only be a constant if $a = c = 0$, or $b = d = 0$, or if $a = c$ and $b = d$. Since we assumed $a$ and $b$ were nonzero, the last case must hold, and thus $\mathcal{P}(\beta) = 1$ for every $\beta \in [\beta^-,\beta^+]$, so the minimizer is obtained in this case, since $\mathcal{P}(\beta) \geq 1$ for all $\beta$.
\end{proof}
In the next section we present a method for finding the minimizer of a function satisfying properties (a) and (b) in Theorem~\ref{thm:pff} to within arbitrary precision. Before doing so, we highlight the importance of ensuring that \emph{both} properties hold. In Figure~\ref{fig:pff} we plot two toy functions, $\pp$ and $\qq$. Although both satisfy property (a), only $\pp$ additionally satisfies (b). Assume we do not have explicit representations of either function, but we can query them at specific points to help find their minimizers. Consider Figure~\ref{fig:pff}. If we query $\pp$ at points $\beta_1$ and $\beta_2$ to find that $\pp(\beta_1) = \pp(\beta_2)$, then choosing any third point $\beta_3 \in (\beta_1,\beta_2)$ will get us closer to the minimizer. However, if $\qq(\beta_1) = \qq(\beta_2)$ for some $\beta_1$, $\beta_2$, we cannot be sure these points are not part of a flat region of $\qq$ somewhere far from the minimizer. It thus becomes unclear how to choose a third point $\beta_3$ at which to query $\qq$. If we choose some $\beta_3 \in (\beta_1,\beta_2)$ and find that $\qq(\beta_3) = \qq(\beta_2) = \qq(\beta_1)$, the minimizer may be within $[\beta_1,\beta_2]$, within $[\beta_2,\beta_3]$, or in a completely different region. Thus it is important for the denominator of a parameter fitness function to be piecewise linear in addition to being concave, since this piecewise linear assumption guarantees property (b) will hold.
\begin{figure}[t]
	\begin{minipage}[b]{0.4\linewidth}
		\centering
		\includegraphics[width=\linewidth]{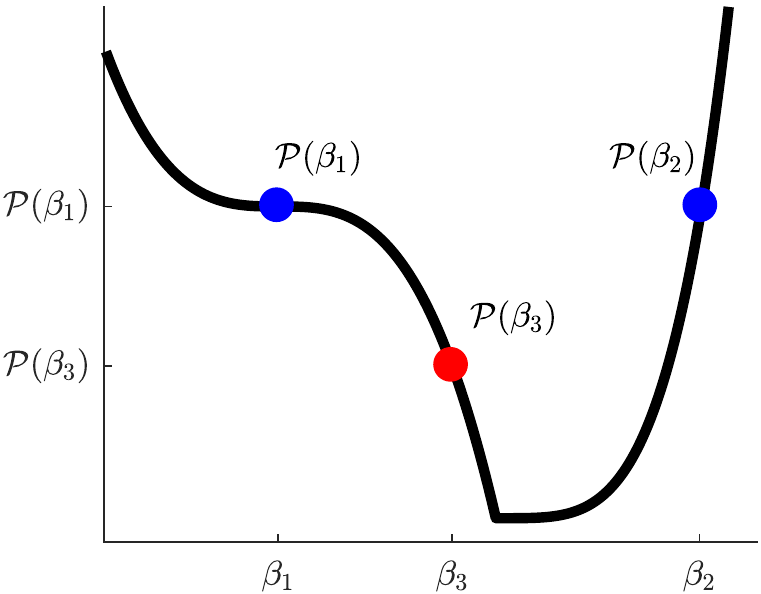}
	\end{minipage}
	\hfill
	\begin{minipage}[b]{0.4\linewidth}
		\centering
		\includegraphics[width=\linewidth]{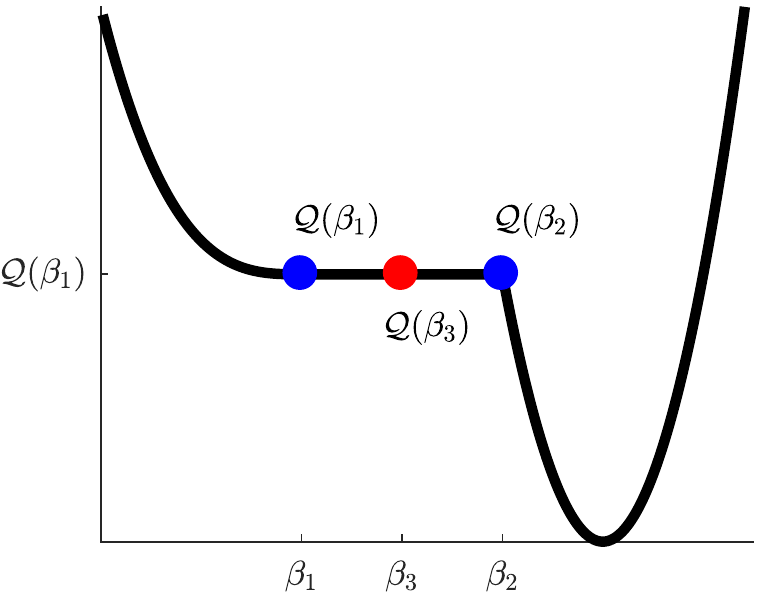}
	\end{minipage}
	\caption{Function $\pp$ satisfies both properties (a) and (b) in Theorem~\ref{thm:pff}. If $\pp(\beta_1) = \pp(\beta_2)$, querying $\pp$ at any point $\beta_3 \in [\beta_1,\beta_2]$ gets us closer to a minimizer. Function $\qq$ only satisfies property (a). If $\qq(\beta_1) = \qq(\beta_2)$, we can get stuck making queries inside a flat region of $\qq$ not near a minimizer.}
	\label{fig:pff}
	\vspace{-.5\baselineskip}
\end{figure}

\subsection{Minimizing $\pp$}
We now outline an approach for finding a minimizer of $\pp$ to within arbitrary precision when Theorem~\ref{thm:pff} holds. Our approach is closely related to the standard bisection method for finding zeros of a continuous function $f$.
Recall that standard bisection starts with $a$ and $b$ such that $sign(f(a)) \neq sign(f(b))$, and then computes $f(c)$ where $c = (a+b)/2$. Checking the sign of $f(c)$ allows one to determine whether the zero of $f$ is located within the interval $[a,c]$ or $[b,c]$. Thus each new query of the function $f$ halves the interval in which a zero must be located.

\begin{figure}[]
	\begin{minipage}[b]{0.4\linewidth}
		\centering
		\includegraphics[width=\linewidth]{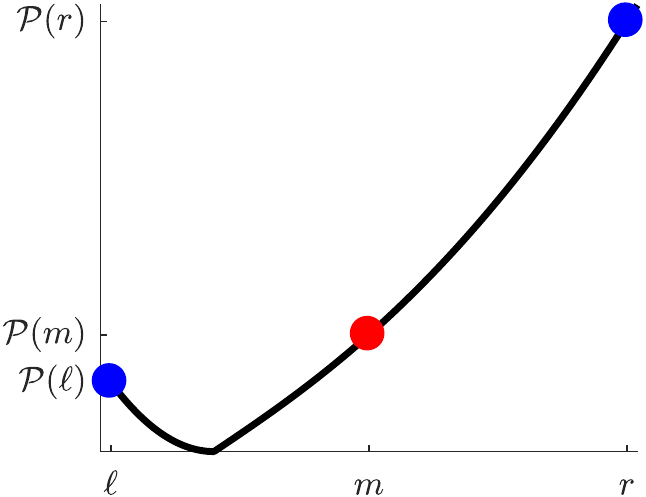}
	\end{minipage}
	\hfill
	\begin{minipage}[b]{0.4\linewidth}
		\centering
		\includegraphics[width=\linewidth]{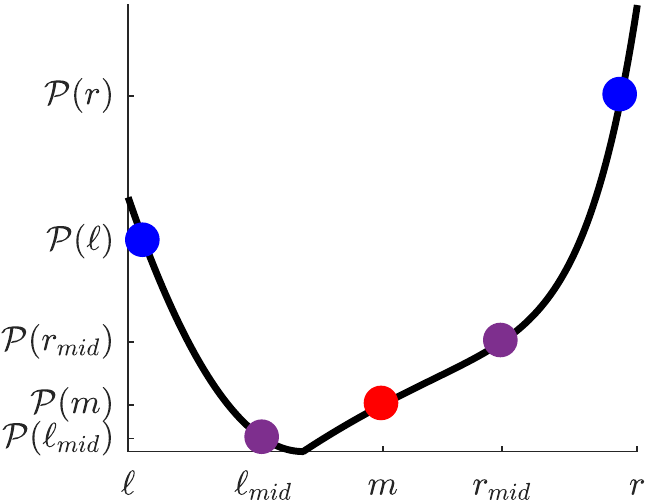}
	\end{minipage}
	\caption{We evaluate $\pp$ at left and right bounds (blue points), and at a midpoint $m$ (red point). Left: If $\mathcal{P}(\ell) < \mathcal{P}(m) < \mathcal{P}(r)$, then we know the minimizer of $\pp$ is in $[\ell,m]$, and we recursively call the \emph{one-branch} phase (Algorithm 1) with new bounds $\ell$ and $m$. 
		Right: If $\mathcal{P}(m) < \mathcal{P}(\ell) < \mathcal{P}(r)$, we don't know if the minimizer is in the left branch $[\ell,m]$ or right branch $[m,r]$. Evaluating $\pp$ at the midpoint of each branch (purple points), we rule out branch $[m,r]$ and recursively call Algorithm 2 with new endpoints $\ell$ and $m$ and midpoint $\ell_{\mathit{mid}}$.
	}
	\label{fig:algs}
			\vspace{-.5\baselineskip}
\end{figure}
Assume $\pp$ satisfies properties (a) and (b) in Theorem~\ref{thm:pff} over an interval $[\ell, r]$.
To satisfy Goal 3, given in~\eqref{eq:best} in Section~\ref{goals}, it suffices to find any minimizer of $\pp$, which we do by repeatedly halving the interval in which the minimizers of $\pp$ must lie.
Our approach differs from standard bisection in that we are trying to find a \emph{minimizer} instead of the zero of some function. The key algorithmic difference is that querying $\pp$ at a single point between two bounds will not always be sufficient to cut the search space in half. Consider Figure~\ref{fig:algs}. Our method starts in a \emph{one-branch} phase in which we know a minimizer lies between $\ell$ and $r$. If we compute $m = (\ell + r)/2$ and find that $\pp(m)$ is between $\pp(\ell)$ and $\pp(r)$, this does in fact automatically cut our search space in half, as this implies that $\pp$ is monotonic on either $[\ell,m]$ or $[m,r]$. However, if $\pp(m) < \min \{\pp(\ell),\pp(r) \}$, then it is possible for the minimizer to reside within either the left branch $[\ell,m]$ or the right branch $[m,r]$. In this case, the method enters a \emph{two-branch} phase in which it takes the midpoint of each branch ($\ell_{\mathit{mid}} = (\ell+m)/2$ and $r_{\mathit{mid}} = (m+r)/2$) and evaluates $\pp(\ell_{\mathit{mid}})$ and $\pp(r_{\mathit{mid}})$. If $\pp$ returns the same value for two of the inputs (e.g., $\pp(\ell) = \pp(m)$), then by property (b) we have found a new interval containing the minimizer(s) of $\pp$ that is at most half the length of $[\ell,r]$. Otherwise, we can use property (a) to deduce that the minimizer will be located within $[\ell, m]$, $[m,r]$, or $[\ell_{\mathit{mid}}, r_{\mathit{mid}}]$, and we recurse on the two-branch phase.

Algorithms~\ref{alg:OneBranch} and~\ref{alg:TwoBranch} handle the one- and two-branch phases of the method respectively. The guarantees of our method are summarized in Theorem~\ref{thm:algs}. We omit the full proof, since it follows directly from considering different simple cases and applying properties of $\pp$ to halve the search space as outline above.
\begin{theorem}
	\label{thm:algs}
	Consider a fixed clustering $C_x$ and a corresponding parameter fitness function $\pp_{C_x}$ satisfying the assumptions of Theorem~\ref{thm:pff}. Running Algorithm~\ref{alg:OneBranch} with input $\ell, r$ and a tolerance $\epsilon$ will produce a resolution parameter $\tilde{\beta}$ that is within $\epsilon$ of the minimizer of $\pp_X$ over the interval $[\ell, r]$, in at most $\log_2 ((r-\ell)/\epsilon)$ recursive calls.
\end{theorem}

\begin{algorithm}
	\caption{CheckOneBranch$(\ell,r,\epsilon)$}
	\begin{algorithmic}[5]
		\State \emph{Base case:}
		\If{$r - \ell < \epsilon$ }
		\State \Return $\ell$
		\EndIf
		\State \emph{Recursive call:}
		\State Midpoint: $m = (\ell+r)/2$
		\Switch{$\ell, m, r$}
		\Case{$\pp(\ell) = \pp(m) = \pp(r)$}
		\State \Return $m$
		\EndCase
		\Case{$\pp(\ell) \leq \pp(m) < \pp(r)$}
		\State \Return CheckOneBranch$(\ell, m, \epsilon)$
		\EndCase
		\Case{$\pp(\ell) > \pp(m) \geq \pp(r)$}
		\State \Return CheckOneBranch$(m, r, \epsilon)$
		\EndCase
		\Case{$\pp(\ell) > \pp(m) < \pp(r)$}
		\State \Return CheckTwoBranches$(\ell,m, r, \epsilon)$
		\EndCase
		\EndSwitch
	\end{algorithmic}
	\label{alg:OneBranch}
\end{algorithm}

\begin{algorithm}[t]
	\caption{CheckTwoBranches$(\ell,m,r,\epsilon)$}
	\begin{algorithmic}[5]
		\State \emph{Base case:}
		\If{$r - \ell < \epsilon$ }
		\State \Return $m$
		\EndIf
		\State \emph{Recursive call:}
		\State Left midpoint: $\ell_{\mathit{mid}} = (\ell+m)/2$
		\State Right midpoint: $r_{\mathit{mid}} = (m+r)/2$
		\Switch{$\ell_{\mathit{mid}}, m, r_{\mathit{mid}}$}
		\Case{$\pp(\ell_{\mathit{mid}}) = \pp(m) = \pp(r_{\mathit{mid}})$}
		\State \Return $m$ 
		\EndCase
		\Case{$\pp(\ell_{\mathit{mid}}) = \pp(m) \neq \pp(r_{\mathit{mid}})$}
		\State \Return CheckOneBranch$(\ell_{\mathit{mid}}, m, \epsilon)$
		\EndCase
		\Case{$\pp(\ell_{\mathit{mid}}) \neq \pp(m) = \pp(r_{\mathit{mid}})$}
		\State \Return CheckOneBranch$(m, r_{\mathit{mid}},\epsilon)$
		\EndCase
		\Case{$\pp(\ell_{\mathit{mid}}) < \pp(m) < \pp(r_{\mathit{mid}})$}
		\State \Return CheckTwoBranches$(\ell,\ell_{\mathit{mid}},m, \epsilon)$
		\EndCase
		\Case{$\pp(\ell_{\mathit{mid}}) > \pp(m) > \pp(r_{\mathit{mid}})$}
		\State \Return CheckTwoBranches$(m, r_{\mathit{mid}},r, \epsilon)$
		\EndCase
		\Case{$\pp(\ell_{\mathit{mid}}) > \pp(m) < \pp(r_{\mathit{mid}})$}
		\State \Return CheckTwoBranches$(\ell_{\mathit{mid}},m, r_{\mathit{mid}}, \epsilon)$
		\EndCase
		\EndSwitch
	\end{algorithmic}
	\label{alg:TwoBranch}
\end{algorithm}

\section{Application to Specific Objectives}
Theorem~\ref{thm:pff} and our approach for minimizing $\pp$ can be immediately applied to learn resolution parameters for the \lcc global clustering objective and the local flow clustering objective. 

\subsection{Local Clustering}
For local clustering we consider the objective function $f_\alpha$ given in~\eqref{stcut} and note that the set of valid clusterings $\mathcal{C}$ is the set of bipartitions. The example clustering we are given at the outset of the problem is $C_x = \{ X, \bar{X} \}$ where $X \subset V$ is some nontrivial set of nodes representing a ``good'' cluster for a given application. We assume we are also given a reference set $R$ (with $\vol(R) \leq \vol(\bar{R})$) that defines a region of the graph in which we are searching for clusters. As noted previously, $f_\alpha$ can be viewed as a parametric linear program, and furthermore it will evaluate to a non-negative finite number for any $\alpha > 0$. Thus by Theorem~\ref{thm:concave}, $G(\alpha) = \min_{S} f_\alpha(S)$ is concave and piecewise linear and we can apply Theorem~\ref{thm:pff}. More explicitly, the local clustering parameter fitness function is
\begin{equation}
\label{localpff}
\pp_X(\alpha) = \frac{\cut(X) + \alpha \vol(\bar{X} \cap R) + \alpha \varepsilon \vol(X \cap \bar{R})}{\min_S [\cut(S) + \alpha \vol(\bar{S} \cap R) + \alpha \varepsilon \vol(S \cap \bar{R})]}.
\end{equation}
If we focus on finding clusters that are subsets of $R$, using objective~\eqref{mqiobj}, we have a simplified fitness function:
\begin{equation}
\label{mqipff}
\pp_{X}(\alpha) = \frac{\cut(X) - \alpha \vol(X) + \alpha \vol(R)}{\min_{S\subseteq R} [\cut(S) - \alpha \vol(S) + \alpha \vol(R)]}.
\end{equation}
When we apply Algorithm~\ref{alg:OneBranch} to minimize~\eqref{localpff} or~\eqref{mqipff}, we can query $\pp_X$ in the time it takes to evaluate a linear function and the time it takes to solve the $s$-$t$ cut problem~\eqref{stcut}. This can be done extremely quickly using localized min-cut computations~\cite{LangRao2004,OrecchiaZhu2014,VeldtGleichMahoney2016,Veldt2019flow}. 


Functions~\eqref{localpff} and~\eqref{mqipff} should be minimized over $\alpha \in [\alpha^*, \cut(R)]$, where $\alpha^*$ is either the minimum of~\eqref{localcond} if we are minimizing~\eqref{localpff}, or is the minimum conductance for a subset of $R$ if we are minimizing~\eqref{mqipff}. One can show that for any $\alpha$ outside this range, objectives~\eqref{stcut} and~\eqref{mqiobj} will be trivially minimized by $S = R$, so it is not meaningful to optimize these objectives for these $\alpha$. In practice one can additionally set stricter upper and lower bounds if desired.

%
 
\subsection{Global Clustering Approach}
We separately consider the standard and degree-weighted versions of \lcc when applying Theorem~\ref{thm:pff} to global graph clustering.

\paragraph{Standard \lcc}
For the standard objective, it is useful to consider the scaled version of \lcc obtained by dividing~\eqref{eq:lcc} by $1-\lambda$ and substituting for a new resolution parameter $\gamma = \lambda/(1-\lambda)$. Then the objective is 
\begin{equation}
\label{eq:gcc}
\min \,\, {\textstyle \sum_{(u,v) \in E} (1-\delta_{uv}) + \sum_{(u,v) \notin E} \gamma \delta_{uv}}.
\end{equation}
The denominator of the parameter fitness function for this scaled $\lcc$ problem would be 
\begin{equation}
\label{gammaCC}
{\textstyle G(\gamma) = \min_{\vx \in \mathcal{X} } \,\, \sum_{(u,v) \in E} x_{uv} + \sum_{(u,v) \notin E}\gamma (1-x_{uv})}
\end{equation}
where $\mathcal{X}$ represents the set of linear constraints for the linear program~\eqref{eq:lccLP}. Note that $G(\gamma)$ will be finite for every $\gamma \geq 0$, so Theorem~\ref{thm:concave} holds. Thus $G$ is concave and piecewise linear as required by Theorem~\ref{thm:pff}. Next, for a fixed clustering $C_x$, let $P_x$ be the number of positive mistakes (pairs of nodes that are separated despite sharing an edge) and $N_x$ be the number of negative mistakes (pairs of nodes that are clustered together but share no edge). Then objective~\eqref{eq:gcc} for this clustering is $P_x + \gamma N_x$, and we see that this fits the linear form given in Theorem~\ref{thm:pff} as long as the example clustering satisfies $P_x > 0$ and $N_x > 0$, which will be the case for nearly any nontrivial clustering one might consider. Finally, note that the parameter fitness function for~\eqref{eq:gcc} would be exactly the same as the parameter fitness function for the standard \lcc objective, since scaling by $(1-\lambda)$ makes no difference if we are going to minimize the ratio between the clustering objective and its LP relaxation. The parameter fitness function for standard \lcc is therefore
\begin{equation}
\label{globalpff}
\pp_{C_x}(\lambda) = \frac{(1-\lambda)P_x + \lambda N_x}{\min_{\vx} \left[ \sum_{uv\in E} (1-\lambda) x_{uv} + \sum_{uv \notin E}\lambda (1-x_{uv})\right] }
\end{equation}
and it satisfies the assumptions of Theorem~\ref{thm:pff} as long as $P_x > 0$, $N_x > 0$, and we optimize over $\lambda \in (0,1)$.

\paragraph{Degree-weighted \lcc}
Showing how Theorem~\ref{thm:pff} applies to degree-weighted \lcc requires slightly more work, though the same basic principles hold. The LP-relaxation of the objective is still a parametric linear program, thus is still concave and piecewise linear in $\lambda$ over the interval $(0,1)$. The denominator of the parameter fitness function in this case would be:
\begin{equation}
\label{eq:deglcc}
{\textstyle \min \,\, \sum_{(u,v) \in E^+} e_{uv}(1-\delta_{uv}) + \sum_{(u,v) \in E^-} e_{uv} \delta_{uv}.}
\end{equation}
where $e_{uv}$ is defined in the degree-weighted fashion (see Section~\ref{sec:lamcc}). For a fixed example clustering $C_x$ encoded by a function $\delta_x = (\delta_{uv})$, we can rearrange this into the form $a + \lambda b$ where $a = \sum_{(u,v) \in E} (1-\delta_{uv})$ and $b =  \sum_{(u,v) \notin E} d_ud_v \delta_{uv} -  \sum_{(u,v) \in E} d_ud_v(1-\delta_{uv})$.
These values are simple to compute, and as long as they are both nonzero, the results of Theorem~\ref{thm:pff} apply. In some extreme cases it is possible that $a = 0$ or $b = 0$, but we expect this to be rare. Furthermore, our general approach may still work even when $a = 0$ or $b = 0$, Theorem~\ref{thm:pff} simply does not analyze this case. We leave it as future work to develop more refined sufficient and necessary conditions such that Algorithm~\ref{alg:OneBranch} is guaranteed to minimize $\pp$.

\section{Experiments}
We consider several local and global clustering experiments in which significant benefit can be gained from learning resolution parameters rather than using previous off-the-shelf algorithms and objective functions. We implement Algorithms~\ref{alg:OneBranch} and~\ref{alg:TwoBranch} in the Julia programming language for both local and global parameter fitness functions. Computing the \lcc linear programming relaxation can be challenging due to the size of the constraint set. For our smaller graphs we apply Gurobi optimization software, and for larger problems we use recently developed memory-efficient projection methods~\cite{veldt2018projection,ruggles2019projection}. For the local-flow objective we use a fast Julia implementation we developed in recent work~\cite{Veldt2019flow}.  Our experiments were run on a machine with two Intel Xeon E5-2690 v4 processors. Code for our experiments and algorithms are available at~\url{https://github/nveldt/LearnResParams}.

\subsection{Learning Parameters for Synthetic Datasets}
\begin{figure}[t!]
	\begin{minipage}[b]{0.53\linewidth}
		\centering
		\includegraphics[width=\linewidth]{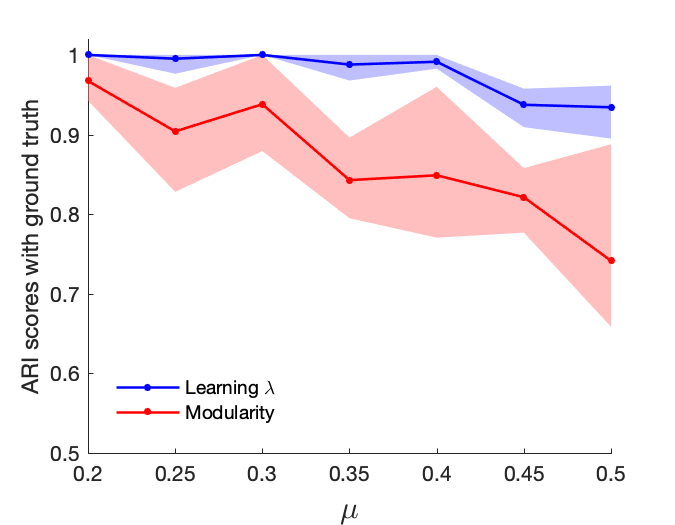}
	\end{minipage}
	\hfill
	\begin{minipage}[b]{0.4\linewidth}
		\centering
		\includegraphics[width=\linewidth]{Figures/mod.pdf}
	\end{minipage}
	\caption{Left: ARI scores for detecting ground truth in LFR graphs. Solid lines indicate mean scores, and colored regions show the range of scores across 5 test graphs for each $\mu$. Right: one of the 5 LFR test graphs for $\mu = 0.3$. Modularity ($\lambda = 1/(2|E|)$) makes mistakes by putting distinct ground truth clusters together (highlighted). For this example our approach perfectly detects the ground truth.}
	\label{fig:lfr}
\end{figure}
Although modularity is a widely-applied objective function for community detection, Fortunato and Barth{\'e}lemy~\cite{Fortunato36} demonstrated that it is unable to accurately detect communities below a certain size threshold in a graph. In our first experiment we demonstrate that learning resolution parameters for \lcc allows us to overcome the \emph{resolution limit} of modularity, and better detect community structure in synthetic networks. We generate a large number of synthetic LFR benchmark graphs~\cite{LFR2009}, in a parameter regime that is chosen to be difficult for modularity. All graphs contain 200 nodes, average degree 10, max degree 20, and community sizes between 5 and 20 nodes. We test a range of mixing parameters $\mu$, which controls the fraction of edges that connect nodes in different communities ($\mu = 0$ means all edges are inside the communities).

For each $\mu$ from $0.2$ to $0.5$, in increments of $0.05$, we generate six LFR networks, one for training and five for testing. On the training graph, we minimize the degree-weighted \lcc parameter fitness function to learn a resolution parameter $\lambda_{\mathit{best}}$. This takes between roughly half an hour (for $\mu = 0.2$) to just over three hours (for $\mu = 0.5$), solving the underlying \lcc LP with Gurobi software. We then cluster the five test LFR examples using a generalized version of the Louvain method~\cite{blondel2008louvain}, as implemented by Jeub et al. ~\cite{genLouvain_software}. We separately run the method with two resolution parameters: $\lambda = 1/(2|E|)$, the standard setting for modularity, and $\lambda = \lambda_\mathit{best}$. Learning $\lambda_\mathit{best}$ significantly improves adjusted Rand index (ARI) scores for detecting the ground truth (see Figure~\ref{fig:lfr}).

\subsection{Local Community Detection}
Next we demonstrate that a small amount of semi-supervised information about target communities in real-world networks can allow us to learn good resolution parameters, leading to more robust community identification. Additionally, minimizing the parameter fitness function provides a way to measure the extent to which \emph{functional} communities in a network correspond to topological notions of community structure in networks.

\paragraph{Data}
We consider four undirected networks, DBLP, Amazon, Orkut, LiveJournal, which are all available on the SNAP repository~\cite{snapnets}, and come with sets of nodes that can be identified as ``functional communities'' (see Yang and Leskovec~\cite{Yang2015}). For example, members of the social network Orkut may explicitly identify as being part of a user-formed group. Such user groups can be viewed simply as metadata about the network, though these still correspond to some notion of community organization that may be desirable to detect. Following an approach taken in previous work~\cite{Veldt2019flow}, we specifically consider the ten largest communities from the 5000 best functional communities as identified by Yang and Leskovec~\cite{Yang2015}. 
The size of each graph in terms of nodes ($n$) and edges $(m)$, along with average set size $|T|$ and conductance $\phi(T)$ among the largest 10 communities, are given in Table~\ref{tab:stats}.

\begin{table}[t]
	\caption{We list the number of nodes ($n$) and edges ($m$) in each snap network, along with average set size $|T|$ and set conductance $\phi(T)$ for the ten largest communities.}
	\label{tab:stats}
	\centering
	\begin{tabular}{lllll}
		\toprule
		Graph & $n$ & $m$  & $|T|$ & $\phi(T)$  \\
		\midrule 
			DBLP & 317,080 & 1,049,866 & 3902 & 0.4948 \\
			Amazon & 334,863 & 925,872 & 190 & 0.0289 \\
			LiveJournal & 3,997,962 & 34,681,189 & 988 & 0.4469 \\
			Orkut & 3,072,441 & 117,185,083 & 3877 & 0.6512 \\
		\bottomrule
	\end{tabular}
\end{table}

\paragraph{Experimental Setup and Results} We treat each functional community as an example cluster $X$. We build a superset of nodes $R$ by growing $X$ from a breadth first search until we have a superset of size $5|X|$, breaking ties arbitrarily. The size of $R$ is chosen so that it comprises a localized region of a large graph, but is still significantly larger than the target cluster $X$ hidden inside of it. We compare two approaches for detecting $X$ within $R$. As a baseline approach we extract the best conductance subset of $R$. Then as our new approach we assume we are given $\cut(X)$ and $\vol(X)$ as additional semi-supervised information. This allows us to minimize the parameter fitness function~\eqref{mqipff}, without knowing what $X$ is. This outputs a resolution parameter $\alpha_X$, and we then minimize $\cut(S) - \alpha_X \vol(S) + \alpha \vol(R)$ over $S \subseteq R$ to output a set $S_X$. 

		\begin{table*}[]
			\caption{
				For experiments on SNAP datasets, we give F1 scores, conductance scores $\phi$, runtimes, and output set sizes for finding the minimum conductance subset ($mc$), and for the set returned by learning a good resolution parameter ($lr$). Display is the average over results for the 10 largest communities in each network.
			}
			\label{tab:snap}
			\centering
			\begin{tabular}{lllllllllllll}
				\toprule
				Graph  & F1 &  & $\phi$ & & run. & & size & \\
				 & $mc$ & $lr$ & $mc$ & $lr$ & $mc$ & $lr$ & $mc$ & $lr$ \\
				\midrule 
DBLP & 0.01 & \textbf{0.47} & 0.02 & 0.16 & 4.9 & 11.4 & 31 & 11680 \\
Amazon & 0.73 & \textbf{0.80} & 0.00 & 0.02 & 0.3 & 0.7 & 142 & 288 \\
LiveJournal & 0.30& \textbf{0.54} & 0.06 & 0.10 & 13.7 & 31.3 & 1556 & 2940 \\
Orkut & 0.44 & \textbf{0.62} & 0.42 & 0.46 & 129.8 & 272.6 & 2353 & 5727 \\
				\bottomrule
			\end{tabular}
		\end{table*}
Table~\ref{tab:snap} reports conductance, set size, runtimes, and F1 scores for both approaches, averaged over the ten communities in each network. Learning resolution parameters leads to significantly better F1 scores on every dataset. 
Additionally, learning resolution parameters for local clustering can be done much more quickly than learning $\lambda$ for \lcc.


\paragraph{New Insights} In addition to improving semi-supervised community detection, minimizing $\pp_X$ allows us to measure how well a functional community matches the topological notion of a cluster. Figure~\ref{fig:snap} shows a scatter plot of F1 community recovery scores against the minimum of $\pp_X$ for each experiment from Table~\ref{tab:snap}. We note a downward sloping trend: small values of $\pp_X$ near 1 tend to indicate that a cluster is highly ``detectable,'' whereas a higher value of $\pp_X$ gives some indication that the functional community may not in fact correspond to a good structural community. We also plot the F1 recovery scores for finding the minimum conductance subset of $R$ against the conductance of functional communities. In this case we do not see any clear pattern, and we learn very little about the relationship between structural and functional communities. 
\begin{figure}[t]
	\begin{minipage}[b]{0.4\linewidth}
		\centering
		\includegraphics[width=\linewidth]{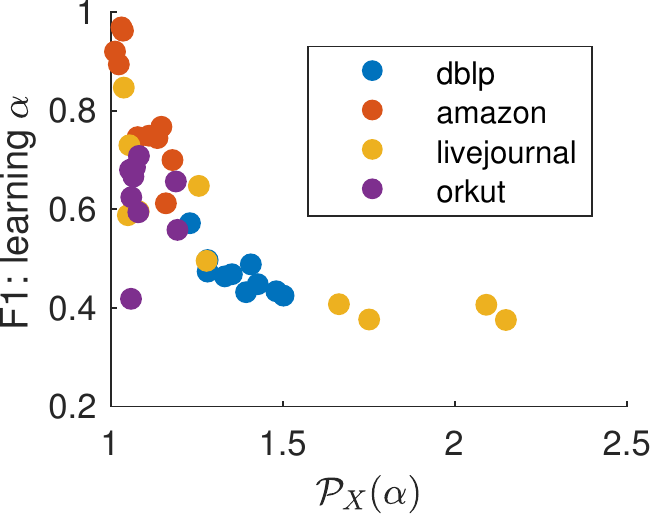}
	\end{minipage}
	\hfill
	\begin{minipage}[b]{0.4\linewidth}
		\centering
		\includegraphics[width=\linewidth]{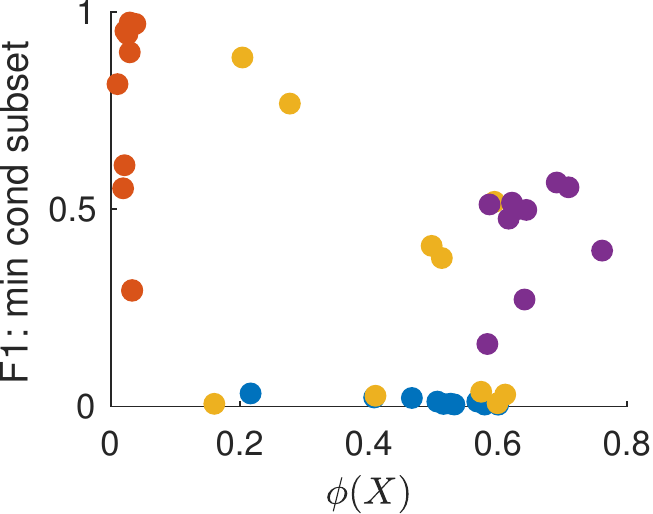}
	\end{minipage}
		\caption{Left: F1 scores for detecting clusters $X$ by learning $\alpha$ vs.\ the minimum of $\pp_X$. Right: F1 scores obtained by finding $\min_{S \subseteq R} \phi(S)$ vs.\ $\phi(X)$. The decreasing, nearly linear pattern in the first plot indicates that the minimum value of $\pp_X$ tells us something about how well the targeted functional communities match a notion of structural communities in a network. The right plot indicates $\phi(X)$ does little to help us predict how detectable a cluster will be.
			}
		\label{fig:snap}
\end{figure}

\subsection{Meta-Data and Global Clustering}
Next we use our techniques to measure how strongly metadata attributes in a network are associated with actual community structure. In general, sets of nodes sharing metadata attributes should not be viewed as ``ground truth'' clusters~\cite{Peel2017ground}, although they may still shed light on the underlying clustering structure of a network.

\paragraph{Email Network} We first consider the largest connected component of an email network~\cite{Leskovec2007graph,Yin2017}. Each of the 986 nodes in the graph represents a faculty member at a European university, and edges represent email correspondence between members. We remove edge weights and directions, and consider an example clustering $C_x$ formed by assigning faculty in the same academic department to the same cluster. We use our bisection method to approximately minimize the global parameter fitness function for the degree-weighted \lcc objective. We run our method until we find the best resolution parameter to within a tolerance of $10^{-8}$, yielding a resolution parameter $\lambda_x = 6.5 \times 10^{-5}$ and a fitness score of $\pp_{C_x}(\lambda_x) = 1.34$. 

To assess how good or bad a score of $1.34$ is for this particular application, we construct a new fake metadata attribute by performing a random permutation of the department labels, which gives a clustering $C_\mathit{fake}$. Approximately minimizing $\pp_{C_\mathit{fake}}$ yields a resolution parameter $\lambda_\mathit{fake} = 3.25\times 10^{-5}$ and a score $\pp_{C_\mathit{fake}}(\lambda_\mathit{fake}) = 2.16$. The gap between the minima of $\pp_{C_\mathit{fake}}$ and $\pp_{C_x}$ indicates that although the true metadata partitioning does not perfectly map to clustering structure in the network, it nevertheless shares some meaningful correlation with the network's connectivity patterns. To further demonstrate this, we run the generalized Louvain algorithm~\cite{blondel2008louvain,genLouvain_software}, using the resolution parameters $\lambda_x$ and $\lambda_\mathit{fake}$. Running the clustering heuristic with $\lambda_x$ outputs a clustering that has a normalized mutual information score (NMI) of 0.71 and an adjusted Rand index (ARI) score of 0.55 with $C_x$. Using $\lambda_\mathit{fake}$, we get NMI and ARI scores of only 0.05 and 0.003 respectively when comparing with $C_\mathit{fake}$. 

\paragraph{Social Networks}
We repeat the above experiment on the smallest social network in the Facebook 100 datasets~\cite{traud2012facebook}, Caltech36. This network is a subset of Facebook with $n = 769$ nodes, defined by users at the California Institute of Technology at a certain point in September 2005. Every node in the network comes with anonymized metadata attributes reporting student/faculty status, gender, major, second major, residence, graduation year, and high school. We treat each metadata attribute as an example clustering $C_x$. Any node with a value of 0 for an attribute we treat as its own cluster, as this indicates the node has no label for the given attribute. We do not run Algorithm~\ref{alg:OneBranch} for each individual $C_x$, since this would involve redundant computations of the \lcc LP relaxations for many of the same values of $\lambda$. Instead, we evaluate the denominator of $\pp$, which is the same for all example clusterings, at 20 equally spaced $\lambda$ values between $1/(8|E|)$ and $2/(|E|)$. We set values of $\lambda$ to be inversely proportional to the number of edges, since we expect the effect of a resolution parameter to depend on a network's size. We note for example that the resolution parameter corresponding to modularity is $\lambda = 1/(2|E|)$, which is also inversely proportional to $|E|$. Computing all of the LP bounds is the bottleneck in our computations, and takes just under 2.5 hours using a recently developed parallel solver for the correlation clustering relaxation~\cite{ruggles2019projection}.

Having evaluated the denominator of $\pp$ at these values, we can quickly find the minimizer of $\pp$ for each metadata attribute and a permuted fake metadata attribute to within an error of less than $10^{-5}$. The smallest values of the parameter fitness function $\pp$ for both real and permuted (fake) metadata attributes are given below:
\[
\begin{tabular}{c c c c c c c c}
& S/F & Gen & Maj. & Maj. 2 & Res. & Yr & HS \\
\midrule
$\min \pp_\mathit{real} $& 1.30 & 1.73 & 2.03 & 2.12 &	1.35 & 1.57 & 2.11 \\
$\min \pp_\mathit{fake} $ & 1.65 & 1.80 & 2.12 & 2.12 & 2.11 & 2.09 & 2.12\\
\bottomrule
\end{tabular}
\]
We note that the smallest values of $\pp$, as well as the largest gap between $\pp$ for true and fake metadata clusterings, are obtained for the student/faculty status, residence, and graduation year attributes. This indicates that these attributes share the strongest correlation with the community structure at this university, which is consistent with independent results on the Facebook 100 datasets~\cite{traud2012facebook,Veldt:2018:CCF:3178876.3186110}.

\subsection{Local Clustering in Social Networks}
In our final experiment we continue exploring the relationship between metadata and community structure in Facebook 100 datasets. We find that minimizing a \emph{local} parameter fitness function $\pp$ can be a much better way to measure the community structure of a set of nodes than simply considering the set's conductance. 


\paragraph{Data} 
We perform experiments on all Facebook 100 networks, focusing on the student/faculty status, gender, residence, and graduation year metadata attributes. For the Caltech dataset in the last experiment, these attained the lowest scores for a \emph{global} parameter fitness function, and furthermore these are the only attributes with a significant number of sets with nontrivial conductance. For the graduation year attribute, we focus on classes between 2006 to 2009, since these correspond to the four primary classes on each campus when the networks were crawled in September of 2005~\cite{traud2012facebook}.

\paragraph{Experimental Setup} We return to an approach similar to our first experiment. For each network and metadata attribute, we consider sets of nodes identified by the same metadata label, e.g., $X$ may represent all students in the class of 2008 at the University of Chicago. We will refer to these simply as \emph{metadata sets}. A label of zero indicates no attribute is known, so we ignore these sets. We also discard sets that are larger than half the graph, or smaller than 20 nodes. We restrict to considering metadata sets with conductance at most 0.7, since conductance scores too close to 1 indicate that a set has little to no meaningful connectivity pattern. For each remaining metadata set $X$, we grow a superset $R$ around $X$ using a breadth first search, and stop growing when $R$ contains half the nodes in the graph or is three times the size of $X$. We then minimize $\pp_X$ as given by~\eqref{mqipff} to learn a resolution parameter $\alpha_X$. This allows us to find $S_X = \argmin_{S \subseteq R} \cut(S) - \alpha_X \vol(S)$, and we then compute the F1 score between $S_X$ and $X$. Our goal here is not to develop a new method for community detection. Rather, computing the F1 score and the minimum of $\pp_X$ provide ways to measure how well a metadata set conforms to a topological notion of community structure, and how detectable the set is from an algorithmic perspective. 

\begin{figure}[t]
	\centering
	\subfloat[$\pp_X$: Gender, S/F, Residence\label{fig:fb1}]
	{\includegraphics[width=.4\linewidth]{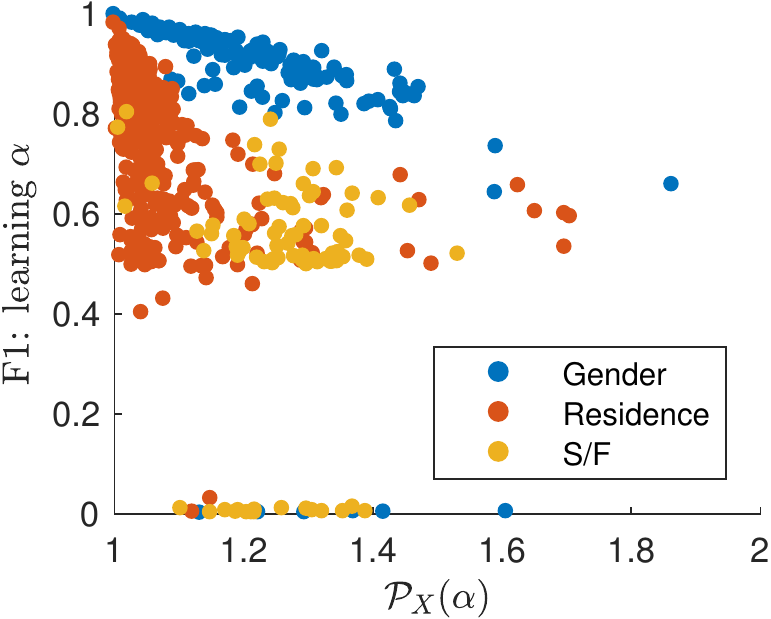}}\hfill
	\subfloat[$\phi(X)$: Gender, S/F, Residence \label{fig:fb2}]
	{\includegraphics[width=.4\linewidth]{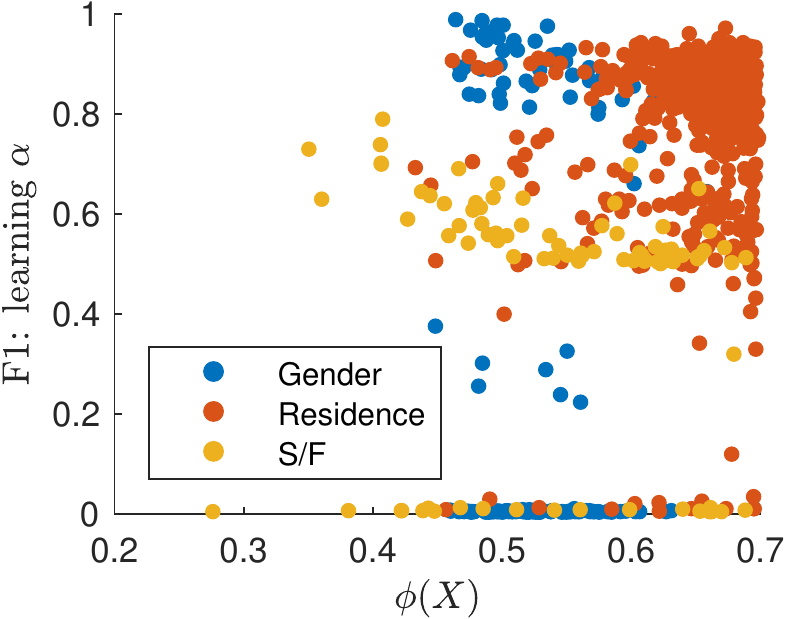}}\hfill
	\centering
	\subfloat[$\pp_X$: Grad Year\label{fig:fb3}]
	{\includegraphics[width=.4\linewidth]{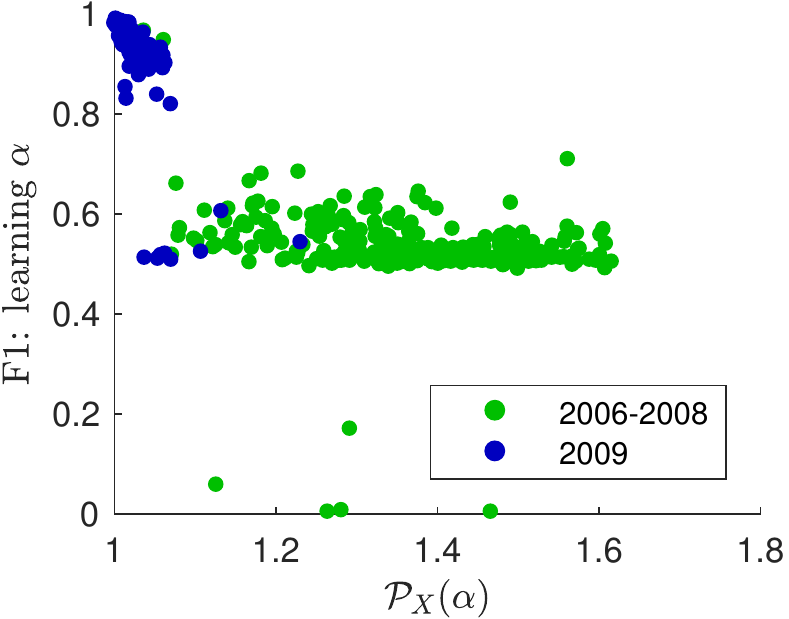}}\hfill
	\subfloat[$\phi(X)$: Grad Year\label{fig:fb4}]
	{\includegraphics[width=.4\linewidth]{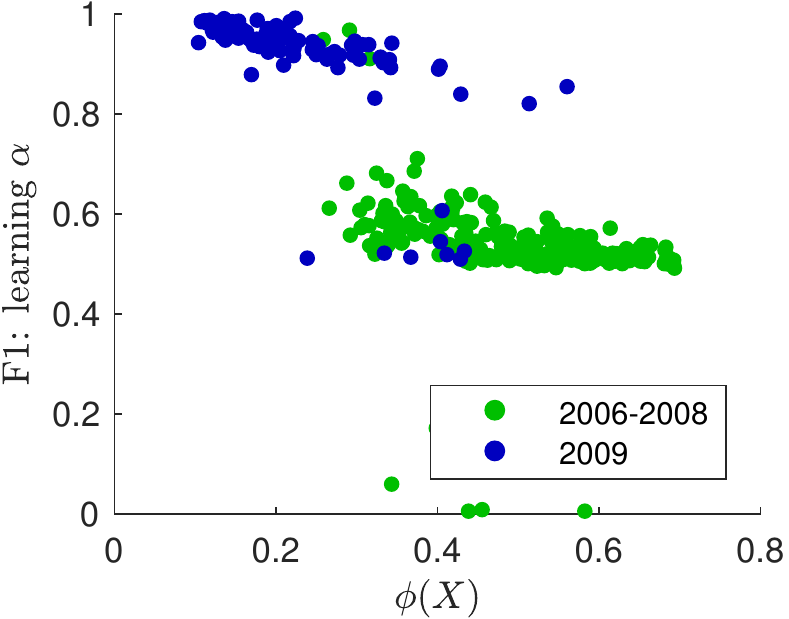}}
	\caption{Minimizing $\pp_X$ gives us refined information about the connectivity structure of different sets sharing metadata attributes in Facebook 100 datasets. Plotting F1 detection scores against the minimum of $\pp_X$ shows especially clear trends for the gender and residence metadata attributes. Plots for the graduation year attribute highlight an anomaly in the connectivity patterns of the 2009 graduating year classes. We explore this in further depth in the main text.}
	\label{fig:fb}
	\vspace{-.5\baselineskip}
\end{figure}
\paragraph{Results} While computing conductance scores provides a good first order measure of a node's community structure, we find that minimizing $\pp$ provides more refined information for the detectability of clusters. In Figure~\ref{fig:fb} we show scatter plots of F1 detection scores against both $\min \pp_X$ as well as $\phi(X)$ for each metadata set $X$. We see that especially for the gender and residence metadata sets across all networks, there is a much clearer relationship between F1 scores and $\min \pp_X$. Values of $\pp_X$ very close to 1 map to F1 scores near 1, and as $\pp_X$ increases we see a downward sloping trend in F1 scores. In the conductance plot we do not see the same trend.

\begin{figure}[t]
	\centering
	\subfloat[Conductance \label{fig:fb5}]
	{\includegraphics[width=.3\linewidth]{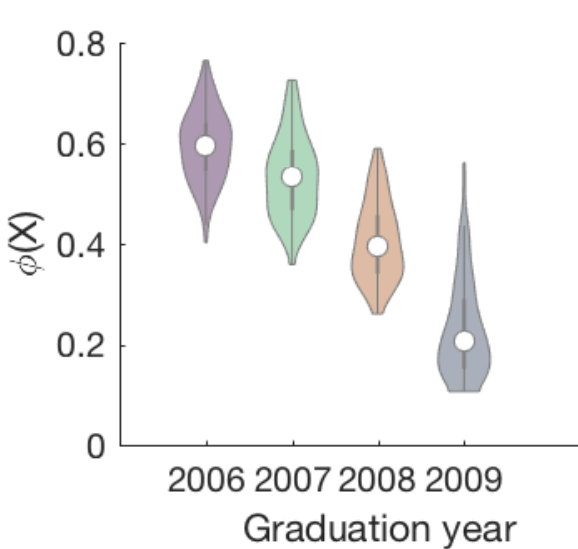}}\hfill
	\subfloat[Volume \label{fig:fb6}]
	{\includegraphics[width=.3\linewidth]{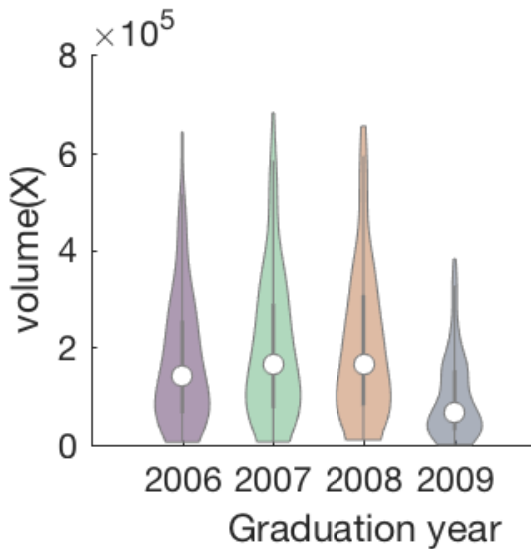}}\hfill
	\centering
	\subfloat[Cut \label{fig:fb7}]
	{\includegraphics[width=.3\linewidth]{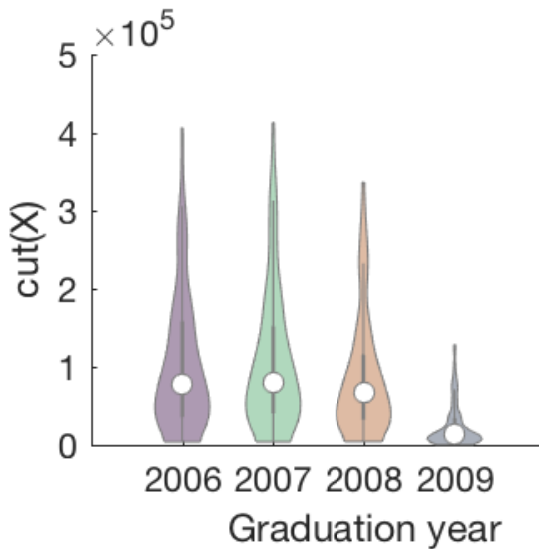}}
	\caption{
		As graduation year increases, conductance scores on the whole tend to decrease. White dots indicate median value. 
		The 2009 graduation year metadata sets tend to be much smaller in volume, but also have very small cut scores, indicating that freshman in 2005 were largely connecting on Facebook with people in their same class.}
	\label{fig:cond}
	\vspace{-.5\baselineskip}
\end{figure}
Figures~\ref{fig:fb3} and~\ref{fig:fb4} show results for metadata sets associated with the 2006-2009 graduation years. For this attribute there appears to be a relationship between both conductance and the $\min \pp$ scores. Furthermore, in both plots we see a separation of the points roughly into two clusters. A deeper exploration of these trends reveals that the 2009 graduation class accounts for the majority of one of these two clusters, and there appears to be an especially clear trend between F1 detection scores and both $\phi(X)$ and $\pp_X$ for this class. 
In order to explain this, we further investigated the connectivity patterns of the main four student classes across all universities. 

\paragraph{New Insights}
Figure~\ref{fig:cond} shows violin plots for $\phi(X)$, $\cut(X)$, and $\vol(X)$ for metadata sets associated with graduation years from 2006 to 2009. Overall, conductance decreases as graduation year increases. We notice that class sizes for the 2009 graduation year are much smaller on average. When these datasets were generated, Facebook users needed a .edu email address to register an account. Thus, in September 2005, the graduation class of 2009 was made up primarily of new freshman who just started college, many of whom had not registered a Facebook account yet. Interestingly, we see a slight decrease in the median cut score from 2007 to 2008, and a significant decrease from 2008 to 2009 (Figure~\ref{fig:fb7}). This suggests that although there were fewer freshman on Facebook at the time, on average they had a greater tendency to establish connections on Facebook among peers in their same graduation year.

Figure~\ref{fig:cond} suggests that in the early months of Facebook, with each new year, students in the same graduating class tended to form tighter Facebook circles with members in their own class. To further explore this hypothesis, for each of the 100 Facebook datasets we consider each node from a graduating class between 2006 and 2009. In each network we compute the average \emph{in-class connection ratio}, i.e., the number of Facebook friends each person has inside the same graduating class, divided by the total number of Facebook connections that the person has across the entire university. In 97 out of 100 datasets (all networks except Caltech36, Hamilton46, and Santa74), this ratio strictly increases as graduation year increases. For Hamilton46 and Santa74, the ratio is still significantly higher for the 2009 graduation class than any other class. If we average this ratio across all networks, as the graduation year increases from 2006 to 2009, the ratios strictly increase: 0.39 for 2006, 0.45 for 2007, 0.57 for 2008, and 0.75 for the class of 2009. In other words, 75\% of an average college freshman's Facebook friends were also freshman, whereas only 39\% of an average senior's Facebook friends were seniors.

Traud et al.~\cite{traud2012facebook} were the first to note the influence of the graduation year attribute on the connectivity structure of Facebook 100 networks. Later, Jacobs et al.~\cite{Jacobs2015} observed differences in the way subgraphs associated with different graduation years evolved and matured over time. These authors noted in particular that the subgraphs associated with the class of 2009 tend to exhibit very skewed degree distributions and comparatively low average degrees. Our observations complement these results, by highlighting heterogeneous behavior in the way members of different classes interacted and connected with one another during the early months of Facebook.



\section{Discussion and Future Work}
We have introduced a new framework and theory for learning resolution parameters based on minimizing a fitness function associated with a single example clustering of interest. 
There are several open questions for improving our specific approach. Our bisection-like algorithm is designed to be general enough to minimize a large class of functions to within arbitrary precision.  However, by making additional assumptions on either specific clustering objectives or the fixed example clustering, one may be able to develop improved algorithms for minimizing the parameter fitness function in practice. Another open question is to study which other graph clustering objectives can fit into out framework, beyond just the \lcc global objective and the local flow clustering objective, and whether, for example, the approach can be applied to clustering in directed graphs.

Our work can be viewed as one approach to the more general goal of learning objective functions for graph clustering applications. This general goal could involve more techniques than simply learning resolution parameters. For example, in future work we wish to explore how to learn small \emph{motif} subgraph patterns~\cite{Benson163} in an example clustering that may be indicative of a desirable type of clustering structure in an application of interest.

\section*{Acknowledgments} The authors thank several funding agencies: Nate Veldt is supported by NSF award IIS-154648, David Gleich is supported by the DARPA SIMPLEX program, the Sloan Foundation, and NSF awards CCF-1149756, CCF-0939370, and IIS-154648. Anthony Wirth is funded by the Australian Research Council.

\bibliographystyle{abbrv}
\bibliography{www19}

\end{document}